\documentclass[12pt]{article}
\usepackage{amsmath}
\usepackage{graphicx,psfrag,epsf}
\usepackage{enumerate}
\usepackage{psfrag,epsf}
\usepackage{url} 
\usepackage{amsmath,amssymb,amsfonts,amsthm,mathtools,mathrsfs}

\usepackage{tikz}
\usepackage{pgfplots}
\usepgfplotslibrary{fillbetween}
\pgfplotsset{compat=1.16}

\usepackage{bm,bbm}
\usepackage{color}
\usepackage{subfigure}
\usepackage{algorithm,algpseudocode}
\usepackage[symbol]{footmisc}
\usepackage{scalefnt}
\usepackage{authblk} 
\usepackage{multirow,centernot}
\usepackage[colorlinks=true, citecolor=blue, urlcolor=blue]{hyperref}
\usepackage{natbib}
\usepackage{dsfont}
\usepackage[title]{appendix}
\usepackage{newtxtext,newtxmath} 
\usepackage{mhchem} 

\input cyracc.def

\usepackage[left=1in,top=1.1in,right=0.5in,bottom=1in]{geometry}

\theoremstyle{definition}

\newtheorem{theorem}{Theorem}[section]

\newtheorem{corollary}[theorem]{Corollary}

\newtheorem{definition}{Definition}[section]
\newtheorem{lemma}[theorem]{Lemma}
\newtheorem{proposition}[theorem]{Proposition}
\newtheorem{remark}[theorem]{Remark}

\makeatletter
\def\@seccntformat#1{\@ifundefined{#1@cntformat}%
	{\csname the#1\endcsname\quad}
	{\csname #1@cntformat\endcsname}
}
\makeatother

\newif\ifShowComments
\ShowCommentstrue
\def\strutdepth{\dp\strutbox}
\def\druk#1{\strut\vadjust{\kern-\strutdepth
        {\vtop to \strutdepth{%
                \baselineskip\strutdepth\vss
                        \llap{\hbox{#1}\quad}\null}}}}

\title{\bf
On the bias of the Hoover index estimator: Results for the gamma distribution
}

\author{
\text{Roberto Vila}$^{1}$\thanks{Corresponding author: Roberto Vila, email: {rovig161@gmail.com}
}
\,\, and
\text{Helton Saulo}$^{1,2}$ 
\\
{\small $^{1}$ Department of Statistics, University of Brasilia, Brasilia, Brazil}\\
{\small $^{2}$ Department of Economics, Federal University of Pelotas, Pelotas, Brazil}\\
}
\begin{document}
	\maketitle 	
	\begin{abstract}
%
%
The Hoover index is a widely used measure of inequality with an intuitive interpretation, yet little is known about the finite-sample properties of its empirical estimator. In this paper, we derive a simple expression for the expected value of the Hoover index estimator for general non-negative populations, based on Laplace transform techniques and exponential tilting. This unified framework applies to both continuous and discrete distributions. Explicit bias expressions are obtained for gamma
populations, showing that the estimator is generally biased in finite samples. Numerical and simulation results illustrate the magnitude of the bias and its dependence on the underlying distribution and sample size.
	\end{abstract}
%

\section{Introduction}

Measures of inequality play a central role in economics, statistics, and the social sciences, providing quantitative tools to summarize disparities in income, wealth, or other non-negative resources \citep{Cowell2011,Sen1973,Atkinson1970}. Among the most commonly used indices, the Gini coefficient has received extensive theoretical and empirical attention, particularly regarding its estimation and finite-sample behavior \citep{Gini1912,Gastwirth1972,Yitzhaki-Schechtman2013}. By contrast, the Hoover index, also known as the Robin Hood, Pietra, or Schutz index, has been comparatively less studied, despite its appealing interpretation as the proportion of total resources that would need to be redistributed to achieve perfect equality \citep{Hoover1936,Schutz1951,Pietra1915}.

	Formally, the Hoover index is defined in terms of the mean absolute deviation from the population mean.
	
	\begin{definition}
		Let $X$ be a non-negative, non-degenerate random variable with finite mean $\mathbb{E}[X]=\mu>0$. The Hoover index \citep{Hoover1936} is defined as
		\begin{align}\label{Hoover-index}
			H = \frac{\mathbb{E}|X-\mu|}{2\mu}.
		\end{align}
	\end{definition}
	
	The Hoover index can therefore be interpreted as one half of the relative mean absolute deviation from the mean, representing the fraction of total income that must be reallocated to achieve perfect equality. 
	Geometrically, the Hoover index corresponds to the maximum vertical distance between the Lorenz curve $L(p)$ and the line of perfect equality, while the Gini coefficient is equal to twice the area enclosed by these two curves (see Figure \ref{fig:gini-hoover}).
	Its formulation places it within the class of $L^1$-type dispersion measures and highlights its close relationship with other inequality indices based on pairwise differences and order statistics \citep{Cowell2011,Yitzhaki1983}.
Its simplicity and intuitive appeal 
have made the Hoover index widely applicable across
diverse fields, including income distribution, regional economics, demography, and environmental studies \citep{Rey-Montouri1999,Haining2003}. However, from a statistical perspective, relatively little is known about the properties of its empirical estimator, especially in finite samples and under specific distributional assumptions.
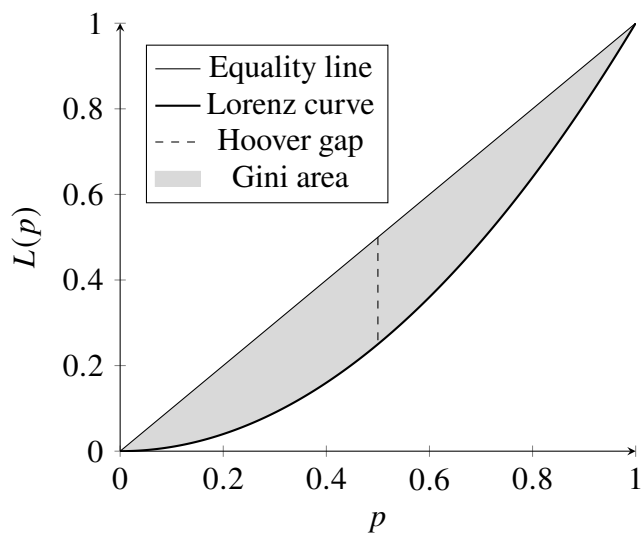
\begin{figure}[H]
	\centering
	\begin{tikzpicture}
		\begin{axis}[
			axis lines=left,
			xlabel={$p$},
			ylabel={$L(p)$},
			xmin=0,xmax=1,
			ymin=0,ymax=1,
			legend style={at={(0.05,0.95)},anchor=north west},
			]
			
			\addplot[name path=A,domain=0:1,samples=200] {x};
			\addlegendentry{Equality line}
			
			\addplot[name path=B,domain=0:1,samples=200, thick] {x^2};
			\addlegendentry{Lorenz curve}

\addplot[dashed] coordinates {(0.5,0.5) (0.5,0.25)};
\addlegendentry{Hoover gap}
			
			\addplot[gray!30] fill between[of=A and B];
			\addlegendentry{Gini area}
			
		\end{axis}
	\end{tikzpicture}
	
	\caption{
		Lorenz curve illustrating the Hoover index and the Gini coefficient.
	}
	\label{fig:gini-hoover}
\end{figure}

The issue of estimator bias is particularly relevant in inequality measurement. It is well documented that the sample Gini coefficient is biased for many common distributions, and a substantial body of literature has been devoted to characterizing and correcting this bias \citep{Deltas2003,Giles2004,Xu2007}. In contrast, analogous results for the Hoover index are scarce, and existing studies are largely limited to asymptotic considerations or numerical investigations.

The main objective of this paper is to fill this gap by providing an exact finite-sample analysis of the Hoover index estimator. We derive a general closed-form representation for its expected value, expressed in terms of Laplace transforms and exponentially tilted distributions. This representation is sufficiently flexible to handle both continuous and discrete populations and highlights structural features such as scale invariance and the role of exponential tilting.

Building on this general result, we derive explicit expressions for the bias of the Hoover index estimator under the gamma distribution. The gamma family is of particular interest because of its widespread use in modeling income, count, and duration data \citep{Johnson1992,Kotz2000}. We conduct a Monte Carlo study to assess the finite-sample behavior of both the original and the bias-corrected estimators. The simulation results indicate that the bias-corrected estimator exhibits markedly improved performance.


The remainder of the paper is organized as follows. Section \ref{The Hoover index: characterizations and special cases} reviews the definition of the Hoover index and presents several useful characterizations for continuous and discrete distributions. Section \ref{Technical results} establishes technical results based on exponential tilting that are instrumental for the main analysis. Section \ref{Bias} derives a simple expression for the expected value of the Hoover index estimator and specializes it to gamma populations. Section \ref{Illustrative simulation study} presents a simulation study illustrating the theoretical findings. Section \ref{Concluding remarks} presents a discussion of the implications of our results and directions for future research. The appendix concludes with the derivation of simple expressions for the bias of the Hoover index estimator for discrete Poisson and geometric distributions.

\section{Characterizations and special cases of the Hoover index}\label{The Hoover index: characterizations and special cases}

This section presents useful characterizations of the Hoover index and derives closed-form expressions for several important special cases, including gamma, Poisson, and geometric distributions.

\begin{proposition}\label{prop-in}
	The Hoover index of a non-negative, non-degenerate random variable $X$ with finite mean $\mathbb{E}[X]=\mu>0$ can be expressed as
	\begin{align*}
H=	
{1\over \mu}
\int_0^\mu
F(t^-)
{\rm d}t,
	\end{align*}
	where $F(t^-)=\mathbb{P}(X<t)$.
\end{proposition}
\begin{proof}
Using the identity
\begin{align}\label{basic-identity}
	\lvert y_1-y_2\rvert
	=
	\int_0^{\infty}
\Big[
\mathds{1}_{[0,t)}(y_1)
+
\mathds{1}_{(0,t)}(y_2)
-
2\,\mathds{1}_{[0,t)}(y_1)\mathds{1}_{(0,t)}(y_2)
\Big]\,
{\rm d}t,
	\quad 
 y_1\geqslant 0, y_2>0.
\end{align}
where $\mathds{1}_A(\cdot)$ denotes the indicator function of a set $A$, we may write
\begin{align*}
	\lvert X-\mu\rvert
	=
	\int_0^{\infty}
	\Big[
	\mathds{1}_{[0,t)}(X)
	+
	\mathds{1}_{(0,t)}(\mu)
	-
	2\,\mathds{1}_{[0,t)}(X)\mathds{1}_{(0,t)}(\mu)
	\Big]\,
	{\rm d}t.
\end{align*}

Taking expectations on both sides and applying Tonelli’s theorem yields
\begin{align*}
	\mathbb{E}\lvert X-\mu\rvert
	=
	\int_\mu^{\infty}
	\big[1-F(t^-)\big]\,
	{\rm d}t
	+
	\int_0^\mu
	F(t^-)\,
	{\rm d}t 
	=
	\int_0^{\infty}
	\big[1-F(t^-)\big]\,
	{\rm d}t
	+
	\int_0^\mu
	[2F(t^-)-1]
	{\rm d}t .
\end{align*}

Finally, noting that
$
\mu
=
\int_0^{\infty}
[1-F(t^-)]\,
{\rm d}t,
$
the desired result follows.
\end{proof}

\begin{remark}
Note that Proposition \ref{prop-in} does not exclude the case $\mathbb{P}(X=0)\in(0,1)$.
\end{remark}

\begin{proposition}\label{pro-2}
The Hoover index of a discrete random variable $X$ supported on $\{0,1,\ldots\}$ with finite mean $\mathbb{E}[X]=\mu>0$
admits the characterization
\begin{align*}
	H
	=
	{1\over \mu}
	\left[ \,
	\sum_{k=0}^{\lfloor \mu\rfloor-1} F(k)
		+(\mu-\lfloor \mu\rfloor)\,F(\lfloor \mu\rfloor)
	\right],
\end{align*}
where $\lfloor t\rfloor$ denotes the floor function.
\end{proposition}
\begin{proof}
Since
\[
F(t^-)=F(\lfloor t\rfloor)=F(k), \qquad t\in[k,k+1),
\]
the function $F(t^-)$ is constant on each interval $[k,k+1)$. It therefore follows that
\begin{align}\label{id-pre}
	\int_{0}^{\mu} F(t^-)\,\mathrm{d}t
	=
	\sum_{k=0}^{\infty}
	\int_{k}^{k+1}
	F(\lfloor t\rfloor)\,
	\mathds{1}_{(t,\infty)}(\mu)\,
	\mathrm{d}t.
\end{align}
Moreover,
\begin{align*}
	\int_{k}^{k+1}
	\mathds{1}_{(t,\infty)}(\mu)\,
	\mathrm{d}t
	=
	\begin{cases}
		1, & k+1\leqslant \mu,\\[2pt]
		\mu-k, & k<\mu<k+1,\\[2pt]
		0, & k\geqslant \mu.
	\end{cases}
\end{align*}
An application of Proposition~\ref{prop-in} together with \eqref{id-pre} completes the proof.
\end{proof}

\begin{proposition}\label{prop-Hoover-index}
The Hoover index of a gamma random variable $X\sim\text{Gamma}(\alpha,\lambda)$
is given by
\begin{align*}
	H
	=
	\frac{\alpha^{\alpha-1}\exp\{-\alpha\}}{\Gamma(\alpha)},
\end{align*}
where $\Gamma(\cdot)$ denotes the (complete) gamma function.
\end{proposition}
\begin{proof}
	A simple calculation shows that
	\begin{align*}
		\int_0^\mu
		F(t^-)
		{\rm d}t
		=
		\int_0^\mu
		{\gamma(\alpha,\lambda t)\over \Gamma(\alpha)}
		{\rm d}t
		=
		\mu
		-
		{1\over\lambda\Gamma(\alpha)}
		\int_0^{\lambda\mu}
		{\Gamma(\alpha,s)}
		{\rm d}s,
\end{align*}
where we have applied the change of variables $s=\lambda t$. Using the identity
 $\int \Gamma(\alpha,u){\rm d}u=u\Gamma(\alpha,u)-\Gamma(\alpha+1,u)+C$, the expression on the right-hand side becomes
\begin{align*}
		\mu
		-
		{1\over\lambda\Gamma(\alpha)}
		\left[\lambda\mu\Gamma(\alpha,\lambda\mu)-\Gamma(\alpha+1,\lambda\mu)
		+
		\Gamma(\alpha+1)
		\right].
	\end{align*}
Finally, by using the identities
$\Gamma(\alpha+1,x)=\alpha\,\Gamma(\alpha,x)+x^{\alpha}\exp\{-x\}$
and
$\Gamma(\alpha+1)=\alpha\,\Gamma(\alpha)$,
together with the relation $\mu=\alpha/\lambda$, the result follows directly from
Proposition \ref{prop-in}.
%
\end{proof}

\begin{remark}
Observe that the quantity $H$ in Proposition~\ref{prop-Hoover-index} can be expressed as
\[
H = f_{\lambda X}(\alpha),
\]
where $f_{\lambda X}$ denotes the probability density function of
$\lambda X \sim \text{Gamma}(\alpha,1)$.

\end{remark}

\begin{proposition}\label{Hoover-Poisson}
The Hoover index of a Poisson random variable $X\sim\text{Poisson}(\lambda)$ is given by
	\begin{align*}
		H=
		{1\over\lambda}
		\left[
		\sum_{k=0}^{\lfloor \lambda\rfloor-1}
			{\Gamma(k+1,\lambda)\over\Gamma(k+1)}+(\lambda-\lfloor\lambda\rfloor)
			{\Gamma(\lfloor \lambda\rfloor+1,\lambda)\over\Gamma(\lfloor \lambda\rfloor+1)}
		\right].
	\end{align*}
\end{proposition}
\begin{proof}
Since $F(k)=\Gamma(k+1,\lambda)/\Gamma(k+1)$ and $\mu=\lambda$, the result follows directly from Proposition \ref{pro-2}.
\end{proof}

\begin{proposition}\label{Hoover-geo}
The Hoover index of a geometric random variable $X\sim\text{Geometric}(p)$,
supported on $\{0,1,\ldots\}$, is given by
	\begin{align*}
	H=p\left(1+\left\lfloor {1-p\over p}\right\rfloor\right)(1-p)^{\lfloor {1-p\over p}\rfloor}.
\end{align*}
\end{proposition}
\begin{proof}
Since $F(k)=1-(1-p)^{k+1}$ and $\mu=(1-p)/p$, and noting that
\[
\sum_{k=0}^{\lfloor \mu\rfloor-1} F(k)
=
\lfloor \mu\rfloor
-
\frac{1-p}{p}\bigl[1-(1-p)^{\lfloor \mu\rfloor}\bigr],
\]
an application of Proposition~\ref{pro-2} completes the proof.
\end{proof}

\section{Technical results}\label{Technical results}

The following technical results play a key role in the development of the results presented in Section \ref{Bias} and are included here as essential intermediate steps.

\begin{lemma}\label{lemma-main}
	If $Y_1$ and $Y_2$ are two independent, non-degenerate, nonnegative random variables, then, 
	for $x_1,x_2>0$,
	\begin{multline*}
		\mathbb{E}[ \vert Y_1-Y_2\vert \exp\{-Y_1 x_1\}\exp\left\{-Y_2 x_2\right\}]
		\\[0,2cm]
		=
		\mathscr{L}_{Y_1}(x_1)\mathscr{L}_{Y_2}(x_2)
		\left\{
		\mathbb{E}[Y_{1,x_1}]
		+
		\mathbb{E}[Y_{2,x_2}]
		-
		2
		\int_0^{\infty}
		[1-F_{Y_{1,x_1}}(t^-)]
		[1-F_{Y_{2,x_2}}(t^-)]
		{\rm d}t
		\right\},
	\end{multline*}
	where  \(\mathscr{L}_{Y_i}(x_i)=\mathbb{E}\left[\exp\{-x_iY_i\}\right]\) denotes the Laplace transform of \(Y_i\), and, for each \(x>0\),
	\begin{align*}
		F_{Y_{i,x_i}}(t)
		=
		{
			\mathbb{E}[
			\mathds{1}_{[0,t]}(Y_i)\exp\{-Y_ix_i\}
			]
			\over 
			\mathscr{L}_{Y_i}(x_i)
		},
		\quad 
		t\geqslant 0,
		\quad i=1,2,
	\end{align*}
	is the cumulative distribution function (CDF) of an exponentially tilted (or Esscher-transformed) random variable \(Y_{i,x}\) \citep{Butler2007}
\end{lemma}
\begin{proof}
	Applying identity~\eqref{basic-identity} and using independence, we obtain
	\begin{align*}
		\mathbb{E}\big[ \lvert Y_1-Y_2\rvert & \exp\{-Y_1 x_1\}\exp\{-Y_2 x_2\}\big]
		\\[0.2cm]
		&=
		\int_0^{\infty}
		\Big\{
		\mathbb{E}\left[
		\mathds{1}_{[0,t)}(Y_1)\exp\{-Y_1x_1\}
		\right]
		\mathscr{L}_{Y_2}(x_2)
		+
		\mathbb{E}\left[
		\mathds{1}_{[0,t)}(Y_2)\exp\{-Y_2x_2\}
		\right]
		\mathscr{L}_{Y_1}(x_1)
		\\[0,2cm]
		&
		-2
		\mathbb{E}\left[
		\mathds{1}_{[0,t)}(Y_1)\exp\{-Y_1x_1\}
		\right]
		\mathbb{E}\left[
		\mathds{1}_{[0,t)}(Y_2)\exp\{-Y_2x_2\}
		\right]
		\Big\}
		\,{\rm d}t .
	\end{align*}
	
	Introducing the distribution functions $F_{Y_{i,x_i}}$, this expression can be rewritten as
	\begin{multline*}
		\mathbb{E}\!\left[ \lvert Y_1-Y_2\rvert \exp\{-Y_1 x_1\}\exp\{-Y_2 x_2\}\right]
		\\[0.2cm]
		=
		\mathscr{L}_{Y_1}(x_1)\mathscr{L}_{Y_2}(x_2)
		\int_0^{\infty}
		\Big\{
		F_{Y_{1,x_1}}(t^-)
		+
		F_{Y_{2,x_2}}(t^-)
		-
		2
		F_{Y_{1,x_1}}(t^-)
		F_{Y_{2,x_2}}(t^-)
		\Big\}
		\,{\rm d}t .
	\end{multline*}
	
	Finally, using the identity
	$
	\mathbb{E}[Y_{i,x_i}]
	=
	\int_{0}^{\infty}[1-F_{Y_{i,x_i}}(t^-)]\,{\rm d}t,
	$
	the result follows.
\end{proof}

\begin{proposition}\label{lemma-main-1}
	Let $Y_1$ and $Y_2$ be two independent, non-negative discrete random variables supported on
	$\{0,1,\ldots\}$. Then, for $x_1,x_2>0$,
	\begin{multline*}
		\mathbb{E}[ \vert Y_1-Y_2\vert \exp\{-Y_1 x_1\}\exp\left\{-Y_2 x_2\right\}]
		\\[0,2cm]
		=
		\mathscr{L}_{Y_1}(x_1)\mathscr{L}_{Y_2}(x_2)
		\left\{
		\mathbb{E}[Y_{1,x_1}]
		+
		\mathbb{E}[Y_{2,x_2}]
		-
		2
		\sum_{k=0}^{\infty}
		[1-F_{Y_{1,x_1}}(k)]
		[1-F_{Y_{2,x_2}}(k)]
		\right\}.
	\end{multline*}
\end{proposition}

\section{Expectation of the Hoover index estimator}\label{Bias}

The main goal of this section is to obtain a simple expression for the expected value of the Hoover index estimator $\widehat{H}$ for population random variables with positive support (see Theorem \ref{main-theorem}). This result then allows us to derive the estimator’s bias in the gamma case (Corollary \ref{main-corollary}). The Hoover index estimator is defined by
\begin{align}\label{gini-estimadtor-def-hoover}
	\widehat{H}
	=
		\displaystyle
	{1\over 2} 
	\left[\dfrac{\displaystyle\sum_{i=1}^{n}\vert X_i-\overline{X}\vert}{\displaystyle\sum_{i=1}^{n} X_i}\right] \mathds{1}_{\left\{\sum_{i=1}^{n} X_i>0\right\}},
	\quad 
	n\in\mathbb{N}, \, n\geqslant 2,
\end{align}
where $\overline{X}=\sum_{i=1}^{n}X_i/n$ is the sample mean and $X_1, X_2,\ldots,X_n$ are independent, identically distributed (i.i.d.) observations from the population $X$.

	\begin{theorem}\label{main-theorem}
Let $X_1, X_2, \ldots$ be independent copies of a non-negative, non-degenerate random variable $X$ with finite mean $\mathbb{E}[X]=\mu>0$. The following holds:
\begin{align*}
		\mathbb{E}\left[\widehat{H}\right]
	=
	{1\over 2}
	\int_0^{\infty}
\mathscr{L}_{X}^{n}(x)
\left\{
\mathbb{E}\left[Y_{1,{x\over n-1}}\right]
+
\mathbb{E}\left[Y_{2,x}\right]
-
2
\int_0^{\infty}
\left[1-F_{Y_{1,{x\over n-1}}}(t^-)\right]
\left[1-F_{Y_{2,x}}(t^-)\right]
{\rm d}t
\right\}
{\rm d}x,
\end{align*}
where \(\mathscr{L}_X(z)=\mathbb{E}\left[\exp\{-zX\}\right]\) denotes the Laplace transform of \(X\).
For each \(x>0\), $Y_{1,{x/(n-1)}}$ and $Y_{2,x}$ are
the exponentially tilted (or Esscher-transformed) \citep{Butler2007} versions of $(n-1)X_1$ and $\sum_{j=2}^n X_j$, respectively; whose CDFs are given by
\begin{align*}
	F_{Y_{1,{x\over n-1}}}(t)
	=
	\frac{\mathbb{E}\left[\exp\{-xX\} \mathds{1}_{[0,{t\over n-1}]}(X)\right]}
	{\mathscr{L}_X(x)},
\quad 
	F_{Y_{2,x}}(t)
	=
	\frac{
		\mathbb{E}\left[\exp\{-x\sum_{j=2}^n X_j\} \mathds{1}_{[0,t]}(\sum_{j=2}^n X_j)\right]
	}
	{\mathscr{L}_X^{n-1}(x)},
	\quad t\geqslant 0,
\end{align*}
respectively.
In the above, we implicitly assume that all Lebesgue-Stieltjes and improper integrals involved are well defined.
	\end{theorem}
\begin{proof}
	By using the well-known identity
$
	{z}
	\int_{0}^{\infty}\exp(-z x){\rm d}x
	=
	1,
$
$z>0$,
with $z=\sum_{i=1}^{n} X_i$, we have
\begin{align}
\mathbb{E}\left[\widehat{H}\right]  
	&=
	\mathbb{E}\left[\widehat{H}\mathds{1}_{\left\{\sum_{i=1}^{n}X_i>0\right\}}\right]
=
	{1\over 2} \,
	\mathbb{E}\left[\sum_{i=1}^{n}\vert X_i-\overline{X}\vert \int_{0}^{\infty}\exp\left\{-\left(\displaystyle\sum_{i=1}^{n} X_i\right) x\right\}{\rm d}x\right]
	\nonumber
	\\[0,2cm]
	&=
	{1\over 2 n}
\sum_{i=1}^{n}
\int_{0}^{\infty}
\mathbb{E}\left[\, 
\left\vert (n-1)X_i-\sum_{\substack{j=1\\ j\neq i}}^n X_j\right\vert \exp\left\{-(n-1)X_i \left({x\over n-1}\right) \right\}\exp\left\{-\left(\displaystyle\sum_{\substack{j=1\\ j\neq i}}^n X_j\right) x\right\}\right]
{\rm d}x,
	\label{eq-1} 
\end{align}
where the last equality follows from Tonelli's theorem, which allows us to commute the order of integration when the integrand is a non-negative measurable function.
Since $X_1,X_2,\ldots$ are i.i.d, it is clear that, for each $i=1,\ldots,n$, the variables $(n-1)X_i$ and $\sum_{\substack{j=1\\ j\neq i}}^n X_j$ are independent, and
\begin{align*}
(n-1)X_i\stackrel{d}{=}(n-1)X_1,
\quad
\sum_{\substack{j=1\\ j\neq i}}^n X_j\stackrel{d}{=}\sum_{j=2}^n X_j,
\end{align*}
where $\stackrel{d}{=}$ denotes equality in distribution. Therefore, from \eqref{eq-1} we have
\begin{align}\label{integral-00}
\mathbb{E}\left[\widehat{H}\right]
	=
	{1\over 2} \,
	\int_{0}^{\infty}
	\mathbb{E}\left[\, 
	\left\vert (n-1)X_1-\sum_{j=2}^n X_j\right\vert 
	\exp\left\{-(n-1)X_1\left({x\over n-1}\right)\right\}
	\exp\left\{-\left(\sum_{j=2}^n X_j\right) x\right\}\right]
	{\rm d}x.
\end{align}

By using Lemma \ref{lemma-main}, we get
\begin{multline}\label{integral-001}
		\mathbb{E}\left[ \, \left\vert (n-1)X_1-\sum_{j=2}^n X_j\right\vert \exp\{-(n-1)X_1 x_1\}\exp\left\{-\sum_{j=2}^n X_j x_2\right\}\right]
\\[0,2cm]
=
\mathscr{L}_{X}^{n}(x)
\left\{
\mathbb{E}\left[Y_{1,{x\over n-1}}\right]
+
\mathbb{E}\left[Y_{2,x}\right]
-
2
\int_0^{\infty}
\left[1-F_{Y_{1,{x\over n-1}}}(t^-)\right]
\left[1-F_{Y_{2,x}}(t^-)\right]
{\rm d}t
\right\},
\end{multline}
where $Y_{1,{x/(n-1)}}$ and $Y_{2,x}$ denote the exponentially tilted versions of $(n-1)X_1$ and $\sum_{j=2}^n X_j$, respectively.

Finally, by combining \eqref{integral-00} and \eqref{integral-001}, the proof follows.
\end{proof}

	\begin{corollary}\label{main-theorem-coro}
	If $X_1, X_2, \ldots$ are independent copies of a non-negative discrete  random variable $X$ with finite mean $\mathbb{E}[X]=\mu>0$, then 
	\begin{align*}
		\mathbb{E}\left[\widehat{H}\right]
		=
		{1\over 2}
		\int_0^{\infty}
		\mathscr{L}_{X}^{n}(x)
		\left\{
		\mathbb{E}\left[Y_{1,{x\over n-1}}\right]
		+
		\mathbb{E}\left[Y_{2,x}\right]
		-
		2
		\sum_{k=0}^{\infty}
		\left[1-F_{Y_{1,{x\over n-1}}}(k)\right]
		\left[1-F_{Y_{2,x}}(k)\right]
		\right\}
		{\rm d}x.
	\end{align*}
	\end{corollary}

\subsection*{Biased Hoover estimator in gamma populations}

Let $X_1,X_2,\ldots$ be independent copies of $X\sim\text{Gamma}(\alpha,\lambda)$.
The Laplace transform of $X$ is given by
\[
\mathscr{L}_X(x)
=
\left(\frac{\lambda}{\lambda+x}\right)^{\alpha}.
\]
Moreover,
\[
(n-1)X_1\sim \text{Gamma}\!\left(\alpha,\frac{\lambda}{\,n-1\,}\right),
\quad
\sum_{j=2}^n X_j \sim \text{Gamma}((n-1)\alpha,\lambda).
\]

Consequently, the exponentially tilted random variables $Y_{1,x/(n-1)}$ and
$Y_{2,x}$ satisfy
\[
Y_{1,{x\over n-1}} \sim \text{Gamma}\!\left(\alpha,\frac{\lambda+x}{\,n-1\,}\right),
\quad
Y_{2,x} \sim \text{Gamma}((n-1)\alpha,\lambda+x).
\]
In particular,
\[
\mathbb{E}\!\left[Y_{1,{x\over n-1}}\right]
=
\mathbb{E}\!\left[Y_{2,x}\right]
=
\frac{(n-1)\alpha}{\lambda+x},
\]
and their distribution functions are given by
\[
F_{Y_{1,{x\over n-1}}}(t)
=
1-
\frac{\Gamma\!\left(\alpha,\frac{\lambda+x}{n-1}\,t\right)}{\Gamma(\alpha)},
\quad
F_{Y_{2,x}}(t)
=
1-
\frac{\Gamma\!\left((n-1)\alpha,(\lambda+x)t\right)}{\Gamma((n-1)\alpha)}.
\]

Applying Theorem~\ref{main-theorem}, we obtain
\begin{align*}
	\mathbb{E}\!\left[\widehat{H}\right]
	&=
	\int_0^{\infty}
	\frac{\lambda^{n\alpha}}{(\lambda+x)^{n\alpha}}
	\Bigg\{
	\frac{(n-1)\alpha}{\lambda+x}
	-
	\int_0^{\infty}
	\frac{\Gamma\!\left(\alpha,\frac{\lambda+x}{n-1}t\right)}{\Gamma(\alpha)}
	\frac{\Gamma\!\left((n-1)\alpha,(\lambda+x)t\right)}{\Gamma((n-1)\alpha)}
	\,\mathrm{d}t
	\Bigg\}
	\mathrm{d}x
	\\[0.2cm]
	&=
	\left[
	\int_0^{\infty}
	\frac{\lambda^{n\alpha}}{(\lambda+x)^{n\alpha+1}}
	\,\mathrm{d}x
	\right]
	\left\{
	(n-1)\alpha
	-
	\int_0^{\infty}
	\frac{\Gamma\!\left(\alpha,\frac{w}{n-1}\right)}{\Gamma(\alpha)}
	\frac{\Gamma\!\left((n-1)\alpha,w\right)}{\Gamma((n-1)\alpha)}
	\,\mathrm{d}w
	\right\},
\end{align*}
where the change of variable $w=(\lambda+x)t$ was used.

Therefore,
\begin{equation}\label{exp-hoover-gamma}
	\mathbb{E}\!\left[\widehat{H}\right]
	=
	\frac{1}{\alpha}
	\left\{
	\left(1-\frac{1}{n}\right)\alpha
	-
	\frac{1}{n}
	\int_0^{\infty}
	\frac{\Gamma\!\left(\alpha,\frac{w}{n-1}\right)}{\Gamma(\alpha)}
	\frac{\Gamma\!\left((n-1)\alpha,w\right)}{\Gamma((n-1)\alpha)}
	\,\mathrm{d}w
	\right\}.
\end{equation}

\begin{remark}\label{remark-ant}
From \eqref{exp-hoover-gamma} we have
\[
0\leqslant \mathbb{E}(\widehat{H})\leqslant 1-\frac{1}{n}.
\]
\end{remark}

\begin{remark}
	Since the  Hoover index estimator $\widehat{H}$ is scale invariant, it is natural that its respective expectation, $\mathbb{E}[\widehat{H}]$, in Item \eqref{exp-hoover-gamma}, does not depend on the rate $\lambda$.
\end{remark}

Therefore, by combining Proposition \ref{prop-Hoover-index} and Item \eqref{exp-hoover-gamma}, we have:
\begin{corollary}\label{main-corollary}
For $n\geqslant 2$, the bias of $\widehat{H}$ relative to $H$, denoted by $\text{Bias}(\widehat{H},H)$, can be written as
\begin{align*}
\text{Bias}(\widehat{H},H)
=
	{1\over \alpha}
\left[
{\left(1-{1\over n}\right)\alpha}
-
{1\over n}
\int_0^\infty
\frac{\Gamma\!\left(\alpha,{w\over n-1}\right)}{\Gamma(\alpha)}
\frac{\Gamma\!\left((n-1)\alpha,w\right)}{\Gamma((n-1)\alpha)}
\,{\rm d}w
\right]
-
{\alpha^{\alpha-1}\exp\{-\alpha\}\over \Gamma(\alpha)}.
\end{align*}
\end{corollary}

\begin{remark}
	Since
	\[
	\int_{0}^{\infty}
	\frac{\Gamma\left(\alpha,\tfrac{w}{n-1}\right)}{\Gamma(\alpha)} \,
	\frac{\Gamma\left((n-1)\alpha,w\right)}{\Gamma((n-1)\alpha)}
	\,\mathrm{d}w
	=
	\mathbb{E}\left[\min\big\{(n-1)U, V\big\}\right],
	\]
	where $
	U \sim \mathrm{Gamma}(\alpha,1)$ and 
	$V \sim \mathrm{Gamma}((n-1)\alpha,1)$
	are independent, we get
	\begin{align*}
		\text{Bias}(\widehat{H},H)
		=
		{1\over \alpha}
		\left\{
		{\left(1-{1\over n}\right)\alpha}
		-
		{1\over n} \,
\mathbb{E}\left[\min\big\{(n-1)U, V\big\}\right]
		\right\}
		-
		{\alpha^{\alpha-1}\exp\{-\alpha\}\over \Gamma(\alpha)}.
	\end{align*}
\end{remark}

\begin{remark}
	From Remark \ref{remark-ant} we have
	\begin{align*}
		-
		{\alpha^{\alpha-1}\exp\{-\alpha\}\over \Gamma(\alpha)}
		\leqslant
\text{Bias}(\widehat{H},H)
\leqslant
1-\frac{1}{n}
-
{\alpha^{\alpha-1}\exp\{-\alpha\}\over \Gamma(\alpha)}.
	\end{align*}
\end{remark}

\begin{proposition}\label{prop-ad-1}
When $n=2$, we have:
\begin{align*}
		\mathbb{E}\left[\widehat{H}\right]
				=
		{1\over 2} \,
		\mathbb{E}\left[\widehat{G}\right]
		=
		{1\over 2} \,
		{\Gamma(\alpha+{1\over 2})\over \sqrt{\pi}\alpha\Gamma(\alpha)}
		=
				{1\over 2} \,G,
\end{align*}
where
\begin{align}\label{gini-estimadtor-def}
	\widehat{G}
	=
	{1\over n-1} 
	\left[\dfrac{\displaystyle\sum_{1\leqslant i<j\leqslant n}\vert X_i-X_j\vert}{\displaystyle\sum_{i=1}^{n} X_i}\right],
	\quad 
	n\in\mathbb{N}, \, n\geqslant 2,
\end{align}
 is the Gini coefficient estimator, proposed by \cite{Deltas2003}, and
\begin{align}\label{Gini-coefficient}
	G={1\over 2}\, {\mathbb{E}\vert X_1-X_2\vert\over\mathbb{E}[X]},
\end{align}
with $X_1$ and $X_2$ being independent copies of $X$, is the (populational) Gini coefficient \citep{Gini1936} of a random variable $X$ with finite mean $\mathbb{E}[X]$.
\end{proposition}
\begin{proof}
	Note that, for $n=2$, we have $\widehat{H}=\widehat{G}/2$. Since $\mathbb{E}[\widehat{G}]=G=	{\Gamma(\alpha+{1/2})/ [\sqrt{\pi}\alpha\Gamma(\alpha)]}$ \cite[see][]{Baydil2025,Vila-Saulo2025}, the result follows immediately.
\end{proof}

\section{Illustrative simulation study}\label{Illustrative simulation study}

This section reports a Monte Carlo study assessing the finite-sample performance of the Hoover
estimator in~\eqref{gini-estimadtor-def-hoover} (uncorrected) and its bias-corrected version obtained by subtracting
an estimate of the finite-sample bias. The correction is implemented in a plug-in fashion using
maximum likelihood (ML) estimation of the distributional parameter.

Let $X_1,\ldots,X_n$ be an i.i.d. sample from a gamma distribution with finite mean.
The uncorrected estimator is given in \eqref{gini-estimadtor-def-hoover}, that is,
\begin{align*}
	\widehat{H}
	=
	\displaystyle
	{1\over 2} 
	\left[\dfrac{\displaystyle\sum_{i=1}^{n}\vert X_i-\overline{X}\vert}{\displaystyle\sum_{i=1}^{n} X_i}\right] \mathds{1}_{\left\{\sum_{i=1}^{n} X_i>0\right\}},
\end{align*}
where $\overline{X}=\sum_{k=1}^{n} X_k / n$ denotes the sample mean. The bias-corrected estimator is defined as
\begin{equation}\label{eq:hoover_hat_corr}
\widehat H^{\,c}
=
\widehat H - \widehat{\mathrm{Bias}}_n(\widehat\alpha),
\end{equation}
where $\widehat\alpha$ is the ML estimator of the shape parameter and
$\widehat{\mathrm{Bias}}_n(\widehat\alpha)$ is obtained by evaluating the analytical bias expression of
$\widehat H$ given in Corollary~4.5 at $\widehat\alpha$ and the given sample size $n$. We consider $X \sim \mathrm{Gamma}(\alpha,1)$ with $\alpha \in \{0.5,1,1.5,2,5\}$. For each Monte Carlo replication and each pair $(n,\alpha)$,
the ML estimator $\widehat\alpha$ is obtained by maximizing the gamma likelihood, and the plug-in
bias correction is computed using the analytical expression in Corollary~4.5. Sample sizes are set to $n \in \{25,50,75,100\}$, and each configuration is replicated $R=2{,}000$ times.
For each replication $r$, we compute $(\widehat H_r,\widehat H^{\,c}_r)$.

Performance is summarized using the relative bias and the root mean squared error (RMSE),
defined respectively as
\[
\mathrm{RelBias}(\widehat H_e)=
\frac{1}{R}\sum_{r=1}^R\frac{\widehat H_{e,r}-H}{H},
\qquad
\mathrm{RMSE}(\widehat H_e)=
\left(\frac{1}{R}\sum_{r=1}^R(\widehat H_{e,r}-H)^2\right)^{1/2},
\]
where $\widehat H_e$ denotes either $\widehat H$ or $\widehat H^{\,c}$ and $H$ is the true Hoover index
under the gamma model.

Figure~\ref{fig:gamma_relbias} presents the corresponding relative bias curves. For all values of
$\alpha$, the uncorrected estimator exhibits a clear negative bias that is more pronounced for
small samples. The bias-corrected estimator effectively removes most of this distortion, leading
to relative bias values that fluctuate around zero even for moderate sample sizes. As expected,
the magnitude of the relative bias decreases with $n$ for both estimators, but the improvement
provided by the bias correction is substantial in small samples.

\begin{figure}[!ht]
\centering
\includegraphics[width=0.90\textwidth]{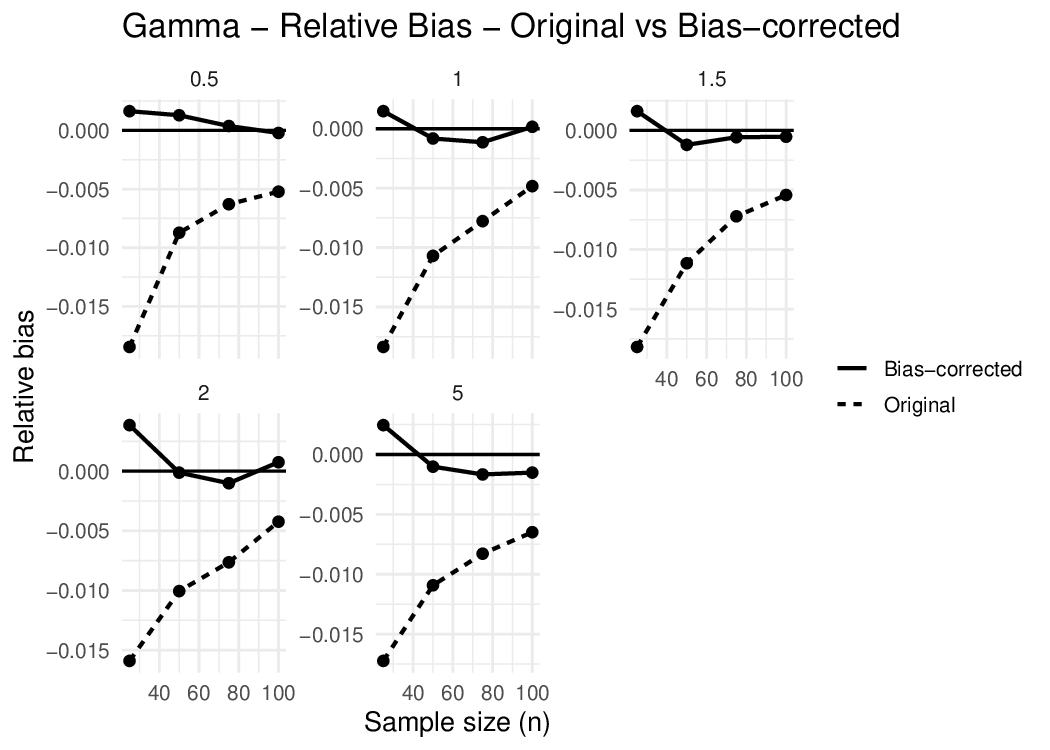}
\caption{Relative bias of the Hoover estimator under the gamma distribution.
Solid line: bias-corrected estimator; dashed line: original estimator.}
\label{fig:gamma_relbias}
\end{figure}

Figure~\ref{fig:gamma_rmse} displays the RMSE of the uncorrected and
bias-corrected Hoover estimators under the gamma model. For all values of the shape parameter
$\alpha$, the RMSE decreases monotonically as the sample size increases, reflecting the
consistency of both estimators. The bias-corrected version exhibits RMSE values that are very
close to those of the original estimator, indicating that the reduction in bias is not achieved
at the expense of a noticeable increase in variability. The two curves
are almost indistinguishable for all sample sizes, which highlights the
effectiveness of the proposed correction in improving accuracy while preserving efficiency.

\begin{figure}[!ht]
\centering
\includegraphics[width=0.90\textwidth]{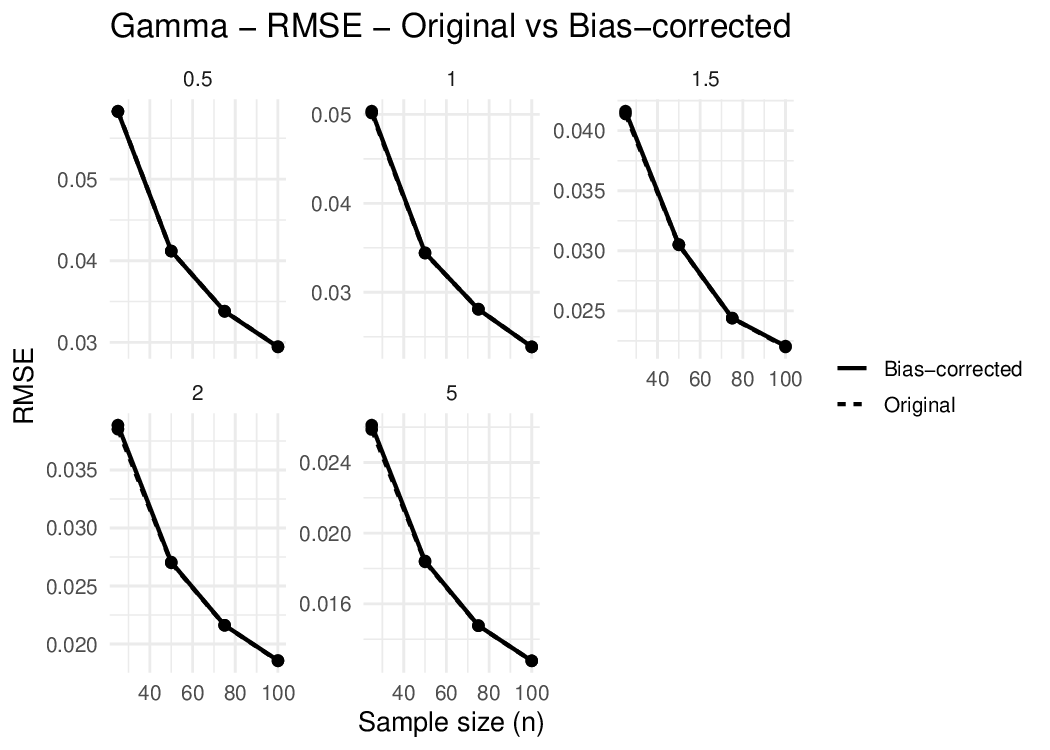}
\caption{RMSE of the Hoover estimator under the gamma distribution.
Solid line: bias-corrected estimator; dashed line: original estimator.}
\label{fig:gamma_rmse}
\end{figure}

\section{Concluding remarks}\label{Concluding remarks}

This paper investigated the finite-sample bias of the Hoover index estimator. By exploiting integral representations based on Laplace transforms and exponential tilting, we derived a general simple expression for its expected value under non-negative populations. Explicit bias formulas were obtained for gamma distribution, showing that the estimator is generally biased in finite samples and that the bias depends on both the sample size and the underlying distribution. The results provide a theoretical basis for numerical bias assessment and suggest the need for bias correction when the Hoover index is used in empirical applications. 
An important direction for future research concerns the relationship between the Hoover index and skewed generalizations of the normal distribution. While the present paper focuses on the sample bias of the Hoover index estimator using exponential tilting and Laplace transform techniques, recent developments indicate that inequality measures may benefit from probabilistic representations based on skew-normal and skew-elliptical families; see \cite{AzzaliniDallaValle1996,BrancoDey2001,Azzalini2014}. In particular, the Gini index has been successfully linked to closed skew-normal and unified skew-elliptical distributions, which has renewed interest in its sampling properties and inferential behavior \cite{CrocettaLoperfido2005,ArellanoValle2014}. A similar connection for the Hoover index appears promising. For instance, when sampling from a bivariate elliptical distribution, the numerator of the Hoover index follows a univariate skew-elliptical distribution, as a direct consequence of results on linear functions of order statistics \cite{Loperfido2008}. In the Gaussian case, the same statistic is expected to follow a closed skew-normal distribution, by arguments analogous to those employed in the study of the Gini index \cite{CrocettaLoperfido2005,DominguezMolina2003}. Establishing such connections may provide new insights into the distributional properties of the Hoover index, potentially enhancing its applicability and relevance in both theoretical and applied contexts. Although these developments fall outside the scope of the present work, they represent a natural and worthwhile avenue for future investigation.


%
%
%



\paragraph*{Acknowledgements}
The research was supported in part by CNPq and CAPES grants from the Brazilian government.

\paragraph*{Disclosure statement}
There are no conflicts of interest to disclose.


\bibliographystyle{apalike}

\appendix
\section{Appendix}
\label{sec:appendix_a}

In this section, we apply Corollary \ref{main-theorem-coro} to derive explicit expressions for the bias of the Hoover estimator when the underlying population is discrete, focusing on the Poisson and geometric distributions. We also highlight that an illustrative simulation study, similar to the one presented for the gamma distribution (see Section \ref{Illustrative simulation study}), can be carried out for these discrete settings as well, providing further numerical insight into the behavior of the estimator.

\subsection{Biased Hoover estimator in Poisson populations}

Consider independent copies $X_1, X_2, \dots$ of $X \sim \text{Poisson}(\lambda)$. Its Laplace transform is
\[
\mathscr{L}_X(x) = \exp\{\lambda (\exp\{-x\} - 1)\}.
\]

For the sum of the remaining variables, we have
\[
\sum_{j=2}^n X_j \sim \text{Poisson}((n-1)\lambda).
\]

Hence, the exponentially tilted random variables $Y_{1,x/(n-1)}$ and $Y_{2,x}$ satisfy
\[
Y_{1,{x\over n-1}} = (n-1) X_{1,x}, \quad \text{where } X_{1,x} \sim \text{Poisson}(\lambda \exp\{-x\}),
\]
and
\[
Y_{2,x} \sim \text{Poisson}((n-1)\lambda \exp\{-x\}).
\]

Both $Y_{1,x/(n-1)}$ and $Y_{2,x}$ share the same expectation:
\[
\mathbb{E}\left[Y_{1,{x\over n-1}}\right] = \mathbb{E}[Y_{2,x}] = (n-1) \lambda \exp\{-x\}.
\]

Their CDFs can be expressed in terms of the lower incomplete gamma function:
\[
F_{Y_{1,{x\over n-1}}}(k) = F_{X_1,x}\left(\frac{k}{n-1}\right) = 1 - \frac{\gamma\left(\left\lfloor {k\over n-1} \right\rfloor + 1, \lambda \exp\{-x\}\right)}{\Gamma\left(\left\lfloor {k\over n-1} \right\rfloor + 1\right)},
\]
and
\[
F_{Y_{2,x}}(k) = 1 - \frac{\gamma(k+1, (n-1)\lambda \exp\{-x\})}{\Gamma(k+1)}.
\]

Applying Corollary \ref{main-theorem-coro}, the expectation of $\widehat{H}$ can be written as
\begin{align}\label{exp-H-P}
	\mathbb{E}\left[\widehat{H}\right]
	&= \int_0^{\infty} \exp\{n \lambda (\exp\{-x\} - 1)\}
	\Bigg[
	(n-1) \lambda \exp\{-x\} \nonumber
	\\[0,2cm]
	&- \sum_{k=0}^{\infty}
	\frac{\gamma\left(\left\lfloor {k\over n-1}\right\rfloor + 1, \lambda \exp\{-x\}\right)}{\Gamma\left(\left\lfloor {k\over n-1}\right\rfloor + 1\right)}
	\frac{\gamma(k+1, (n-1)\lambda \exp\{-x\})}{\Gamma(k+1)}
	\Bigg] {\rm d}x \nonumber
	\\[0,2cm]
	&= \exp\{-n \lambda\} \int_0^{\lambda} \frac{\exp\{n w\}}{w}
	\Bigg[
	(n-1) w - \sum_{k=0}^{\infty}
	\frac{\gamma\left(\left\lfloor {k\over n-1}\right\rfloor + 1, w\right)}{\Gamma\left(\left\lfloor {k\over n-1}\right\rfloor + 1\right)}
	\frac{\gamma(k+1, (n-1) w)}{\Gamma(k+1)}
	\Bigg] {\rm d}w,
\end{align}
where the last expression follows from the change of variable $w = \lambda \exp\{-x\}$.

\begin{remark}\label{remark-ant-2-1}
	From \eqref{exp-H-P} we have
	\begin{align*}
		0\leqslant \mathbb{E}\left[\widehat{H}\right]\leqslant \left(1-{1\over n}\right)
		[1-\exp\{-n\lambda\}].
	\end{align*}
\end{remark}

By combining Proposition \ref{Hoover-Poisson} and Item \eqref{exp-H-P}, we have:
\begin{corollary}\label{main-corollary-2}
	For $n\geqslant 2$, the bias of $\widehat{H}$ relative to $H$, denoted by $\text{Bias}(\widehat{H},H)$, can be written as
	\begin{align*}
		\text{Bias}(\widehat{H},H)
		&=
		\exp\{-n\lambda\}
		\int_0^{\lambda}
		{\exp\{nw\}\over w}
		\Bigg\{
		(n-1)w
		-
		\sum_{k=0}^{\infty}
		{\gamma\left(\left\lfloor{k\over n-1}\right\rfloor+1,w\right)\over\Gamma\left(\left\lfloor{k\over n-1}\right\rfloor+1\right)}
		\,
		{\gamma(k+1,(n-1)w)\over\Gamma(k+1)}
		\Bigg\}
		{\rm d}w
		\\[0,2cm]
		&-
		{1\over\lambda}
		\left[
		\sum_{k=0}^{\lfloor \lambda\rfloor-1}
		{\Gamma(k+1,\lambda)\over\Gamma(k+1)}+(\lambda-\lfloor\lambda\rfloor)
		{\Gamma(\lfloor \lambda\rfloor+1,\lambda)\over\Gamma(\lfloor \lambda\rfloor+1)}
		\right].
	\end{align*}
\end{corollary}

\begin{remark}
	From Remark \ref{remark-ant-2-1} we have
	\begin{multline*}
		-
		{1\over\lambda}
		\left[
		\sum_{k=0}^{\lfloor \lambda\rfloor-1}
		{\Gamma(k+1,\lambda)\over\Gamma(k+1)}+(\lambda-\lfloor\lambda\rfloor)
		{\Gamma(\lfloor \lambda\rfloor+1,\lambda)\over\Gamma(\lfloor \lambda\rfloor+1)}
		\right]
		\\[0,2cm]
		\leqslant
		\text{Bias}(\widehat{H},H)
		\leqslant
		\left(1-{1\over n}\right)
		[1-\exp\{-n\lambda\}]
		-
		{1\over\lambda}
		\left[
		\sum_{k=0}^{\lfloor \lambda\rfloor-1}
		{\Gamma(k+1,\lambda)\over\Gamma(k+1)}+(\lambda-\lfloor\lambda\rfloor)
		{\Gamma(\lfloor \lambda\rfloor+1,\lambda)\over\Gamma(\lfloor \lambda\rfloor+1)}
		\right].
	\end{multline*}
\end{remark}

\subsection{Biased Hoover estimator in geometric populations}

Let $X_1, X_2, \dots$ be independent copies of $X \sim \text{Geometric}(p)$. Its Laplace transform is
\[
\mathscr{L}_X(x) = \frac{p}{1 - (1-p)\exp\{-x\}}.
\]

Moreover, the sum of the remaining variables satisfies
\[
\sum_{j=2}^{n} X_j \sim \text{NegBin}(n-1, p).
\]

Consequently, the exponentially tilted random variables $Y_{1,x/(n-1)}$ and $Y_{2,x}$ are
\[
Y_{1, {x\over n-1}} = (n-1) X_{1,x}, \quad \text{where } X_{1,x} \sim \text{Geometric}(1 - (1-p)\exp\{-x\}),
\]
and
\[
Y_{2,x} \sim \text{NegBin}(n-1, 1 - (1-p)\exp\{-x\}).
\]

In particular, their expectations coincide:
\[
\mathbb{E}\left[Y_{1,{x\over n-1}}\right] = \mathbb{E}\big[Y_{2,x}\big] = (n-1) \frac{(1-p)\exp\{-x\}}{1 - (1-p)\exp\{-x\}}.
\]

Their distribution functions can be expressed as
\[
F_{Y_{1,{x\over n-1}}}(k) = F_{X_1,x}\left(\frac{k}{n-1}\right) = 1 - \left[(1-p) \exp\{-x\}\right]^{\left\lfloor {k\over n-1} \right\rfloor + 1},
\]
and
\[
F_{Y_{2,x}}(k) = 1 - I_{(1-p)\exp\{-x\}}(k+1, n-1),
\]
where $I_p(r,s)$ denotes the regularized incomplete beta function.

Applying Corollary \ref{main-theorem-coro}, the expected value of $\widehat{H}$ can be written as
\[
\begin{aligned}
	\mathbb{E}\left[\widehat{H}\right]
	&= \int_0^{\infty} \frac{p^n}{\big[1-(1-p)\exp\{-x\}\big]^n}
	\Bigg\{
	(n-1)\frac{(1-p)\,\exp\{-x\}}{1-(1-p)\exp\{-x\}}
	\\[0,2cm]
	&- \sum_{k=0}^{\infty} \big[(1-p)\exp\{-x\}\big]^{\left\lfloor {k\over n-1} \right\rfloor + 1}
	I_{(1-p)\exp\{-x\}}(k+1,n-1)
	\Bigg\} {\rm d}x
	\\[0,2cm]
	&= \int_0^{1-p} \frac{p^n}{w(1-w)^n}
	\Bigg\{
	(n-1)\frac{w}{1-w} - \sum_{k=0}^{\infty} w^{\left\lfloor {k\over n-1} \right\rfloor + 1} I_w(k+1,n-1)
	\Bigg\} {\rm d}w,
\end{aligned}
\]
where the last equality follows from the change of variable $w = (1-p)\exp\{-x\}$.

Finally, this expression can also be rewritten as
\begin{align}\label{exp-hoover-gamma-1}
	\mathbb{E}\left[\widehat{H}\right] = \left(1 - \frac{1}{n}\right)(1-p^n) - p^n \int_0^{1-p}
	\sum_{k=0}^{\infty} \frac{w^{\left\lfloor {k\over n-1} \right\rfloor} I_w(k+1,n-1)}{(1-w)^n} \, {\rm d}w.
\end{align}

\begin{remark}\label{remark-ant-2}
	From \eqref{exp-hoover-gamma-1} we have
	\begin{align*}
		0
		\leqslant
		\mathbb{E}\left[\widehat{H}\right]
		\leqslant
		\left(1-{1\over n}\right)(1-p^n).
	\end{align*}
\end{remark}

By combining Proposition \ref{Hoover-geo} and Item \eqref{exp-hoover-gamma-1}, we have:
\begin{corollary}\label{main-corollary-1}
	For $n\geqslant 2$, the bias of $\widehat{H}$ relative to $H$, denoted by $\text{Bias}(\widehat{H},H)$, can be written as
	\begin{align*}
		\text{Bias}(\widehat{H},H)
		=
		\left(1-{1\over n}\right)(1-p^n)
		-
		p^n
		\int_0^{1-p}
		\sum_{k=0}^{\infty}
		\frac{w^{\left\lfloor{k\over n-1}\right\rfloor}
			I_{w}(k+1,n-1)}{(1-w)^n} \,
		{\rm d}w
		-
		p\left(1+\left\lfloor {1-p\over p}\right\rfloor\right)(1-p)^{\lfloor {1-p\over p}\rfloor}.
	\end{align*}
\end{corollary}

\begin{remark}
	From Remark \ref{remark-ant-2} we have
	\begin{multline*}
		-
		p\left(1+\left\lfloor {1-p\over p}\right\rfloor\right)(1-p)^{\lfloor {1-p\over p}\rfloor}
		\leqslant
		\text{Bias}(\widehat{H},H)
		\leqslant
		\left(1-{1\over n}\right)(1-p^n)
		-
		p\left(1+\left\lfloor {1-p\over p}\right\rfloor\right)(1-p)^{\lfloor {1-p\over p}\rfloor}.
	\end{multline*}
\end{remark}

\begin{remark}
	Note that the integrals in Corollaries \ref{main-corollary}, \ref{main-corollary-2}, and \ref{main-corollary-1} lack closed-form representations in terms of standard special functions, but can be efficiently computed numerically.
\end{remark}


%


\end{document}